\documentclass[11pt]{article}

\usepackage[procnumbered,ruled,vlined, algo2e]{algorithm2e}
\usepackage{todonotes}
\usepackage{ctable}	
\usepackage{amsthm}
\usepackage{amsmath, amssymb, amsfonts, verbatim}

\usepackage[margin=1in]{geometry}

\newtheorem{theorem}{Theorem}
\newtheorem{lemma}{Lemma}[section]

\newtheorem{claim}[lemma]{Claim}

\newtheorem{definition}{Definition}

\title{Improved Bounds for Matching in Random-Order Streams} 

\author{
	Aaron Bernstein \footnote{bernstei@gmail.com. This work was done while funded by NSF CAREER Grant 1942010 and Simons Collaboration on Algorithms and Geometry.} \\ 
	Rutgers University 
}

\newcommand{\stream}[2]{S_{[#1, #2]}}
\DeclareMathOperator*{\ex}{\mathbb{E}}
\newcommand{\eps}{\epsilon}

\newcommand{\eearly}{E^{early}}
\newcommand{\elate}{E^{late}}
\newcommand{\gearly}{G^{early}}
\newcommand{\glate}{G^{late}}

\newcommand{\cond}{\ | \ }
\newcommand{\polylog}{\textrm{polylog}}
\newcommand{\poly}{\textrm{poly}}
\newcommand{\foundu}{{\sc FoundUnderfull}}

\newcommand{\ignore}[1]{}
\newcommand{\true}{True}
\newcommand{\false}{False}

\usepackage{algorithm2e}

\begin{document}

\maketitle

\begin{abstract}
We study the problem of computing an approximate maximum cardinality matching in the semi-streaming model when edges arrive in a \emph{random} order. In the semi-streaming model, the edges of the input graph $G = (V,E)$ are given as a stream $e_1, \ldots, e_m$, and the algorithm is allowed to make a single pass over this stream while using $O(n\polylog(n))$ space ($m = |E|$ and $n = |V|$). If the order of edges is adversarial, a simple single-pass greedy algorithm yields a $1/2$-approximation in $O(n)$ space; achieving a better approximation in adversarial streams remains an elusive open question.

A line of recent work shows that one can improve upon the $1/2$-approximation if the edges of the stream arrive in a random order. The state of the art for this model is two-fold: Assadi et al. [SODA 2019] show how to compute a $\frac{2}{3}$$(\sim.66)$-approximate matching, but the space requirement is $O(n^{1.5}\polylog(n))$. Very recently, Farhadi et al. [SODA 2020] presented an algorithm with the desired space usage of $O(n\polylog(n))$, but a worse approximation ratio of $\frac{6}{11}$$(\sim.545)$, or $\frac{3}{5}$$(=.6)$ in bipartite graphs.

In this paper, we present an algorithm that computes a $\frac{2}{3}(\sim.66)$-approximate matching using only $O(n\log(n))$ space, improving upon both results above. We also note that for adversarial streams, a lower bound of Kapralov [SODA 2013] shows that any algorithm that achieves a $1-\frac{1}{e}$($\sim.63$)-approximation requires $(n^{1+\Omega(1/\log\log(n))})$ space. Our result for random-order streams is the first to go beyond the adversarial-order lower bound, thus establishing that computing a maximum matching is provably easier in random-order streams.
\end{abstract} 

\newpage 

\section{Introduction}
Computing a maximum cardinality matching is a classical problem in combinatorial optimization, with a large number of algorithms and applications. Motivated by the rise of massive graphs, much of the recent research on this problem has focused on sub-linear algorithms that are able to compute a matching without storing the entire graph in memory. One of the standard sub-linear models for processing graphs is known as the \emph{semi-streaming} model \cite{FeigenbaumKMSZ05}: the algorithm has access to a sequence of edges (the stream), and is allowed to make a single pass over this sequence while only using only $O(n\polylog(n))$ internal memory, where $n$ is the number of vertices in the graph. Note that the memory used is still significantly smaller than the number of edges in the graph, and that $O(n)$ memory is also necessary if we want the algorithm to output the actual edges of the matching. (One typically assume $O(\log(n))$-size words, so that a single edge can be stored in $O(1)$ space; if one were to express space in terms of the number of bits, all the space bounds in this paper would increase by a $O(\log(n))$ factor.)

If the edges of the stream arrive in an arbitrary order, a simple greedy algorithm can compute a maximal matching -- and hence a $1/2$-approximate maximum matching -- in a single streaming pass and $O(n)$ space. Going beyond a $1/2$-approximation with a single pass is considered one of the main open problems in the area. The strongest lower bound is by Kapralov \cite{Kapralov13}, who build upon an earlier lower bound of Goel et al. \cite{GoelKK12}: any algorithm with approximation ratio $ \geq 1-1/e \sim .63$ requires $n^{1+\Omega(1/\log\log(n))})$ space \cite{Kapralov13}. But we still do not know where the right answer lies between $1/2$ and $1-1/e$. 

To make progress on this intriguing problem, several recent papers studied a more relaxed model, where the graph is still arbitrary, but the edges are assumed to arrive in a \emph{uniformly random} order. Konrad et al. were the first to go beyond a 1/2-approximation in this setting: they showed that in random-order streams, there exists an $O(n)$-space algorithm that computes an .5003-approximate matching, or .5005-approximate for bipartite graphs \cite{KonradMM12}. This was later improved to .506 in general graphs \cite{GamlathKMS19} and .539 in bipartite graphs \cite{Konrad18}. Assadi et al. then showed an algorithm with an approximation ratio of $(2/3 - \eps) \sim .66$, but their algorithm had a significantly larger space requirement of $O(n^{1.5}\polylog(n))$ \cite{AssadiBBMS19}. Finally, very recently (SODA 2020), Farhadi et al. achieved the current state of the art for $O(n\polylog(n))$ space; their algorithm achieves an approximation ratio of $6/11 \sim .545$ for general graphs and $3/5 = .6$ for bipartite graphs \cite{FarhadiHMRR20}. A summary of these results can be found in Table \ref{table}.

\begin{table}
\centering
\setlength{\tabcolsep}{10pt}
\begin{tabular}{l llc} 
\toprule
& \multicolumn{2}{c}{\textsc{Approximation Factor}} & \\
& \textbf{Bipartite graphs} & \textbf{General graphs} & \textbf{Space} \\

\midrule

Konrad et al. \cite{KonradMM12} & $0.5005$ & $0.5003$ & $O(n)$ \\

Gamlath et al. \cite {GamlathKMS19} & $0.512$ & $0.506$ & $O(n\cdot \polylog(n))$ \\
 
Konrad \cite{Konrad18} & $0.539$ & - & $O(n\cdot\polylog(n))$ \\

Assadi et al. \cite{AssadiBBMS19} & $0.666$ & $0.666$ & $O(n^{1.5}\cdot\polylog(n))$ \\

Farhadi et al. \cite{FarhadiHMRR20} & $0.6$ & $0.545$ & $O(n\cdot\polylog(n))$ \\

\textbf{This paper} & $\mathbf{0.666}$ & $\mathbf{0.666}$ & $O(n\log(n))$ \\

\bottomrule
\end{tabular}
\caption{Single-pass semi-streaming algorithms known for the maximum matching when edges arrive in a random order. The space bounds are expressed in terms of $O(\log(n))$-size words, though many existing results do not state the exact $\polylog(n)$ term. The result of Gamlath et al. \cite{GamlathKMS19} works in weighted graphs; all others are restricted to unweighted graphs.}
\label{table}
\vspace{-0.1cm}
\end{table}

Although this line of work suggests that computing a maximum matching might be fundamentally easier in random-order streams, we note that even in bipartite graphs, none of the previous results go beyond the best known lower bound for adversarial streams mentioned above \cite{Kapralov13}: the algorithm of Assadi et al. uses too much space $(n^{1.5} \gg n^{1+1/\log\log(n)})$, while the result of Farhadi et al. has an approximation ratio of $.6 < 1-1/e$. 

Our result is the first to go beyond the adversarial-order lower bound, thus establishing that computing a matching is provably easier in random-order streams.

\begin{theorem} [Our Result]
\label{thm:main}
Given any (possibly non-bipartite) graph $G$ and any approximation parameter $1 > \eps > 0$, there exists a deterministic single-pass streaming algorithm that with high probability computes a $(2/3 - \eps)$-approximate matching if the edges of $G$ arrive in a uniformly random order. The space usage of the algorithm is $O(n\log(n)\poly(\eps^{-1}))$
\end{theorem}

Our result significantly improves upon the space requirement of Assadi et al. \cite{AssadiBBMS19} and the approximation ratio of Farhadi et al. \cite{FarhadiHMRR20}. In fact, our algorithm achieves the best of both those results (see Table \ref{table}). On top of that, our result is quite simple; given that it improves upon a sequence of previous results, we see this simplicity as a plus. 
\paragraph*{Related Work}
If the only requirement is to return an approximate estimate of the \emph{size} of the maximum matching, rather than the actual edges, a surprising result by Kapralov et al. shows that one can get away with very little space: given a single pass over a random-order stream, it is possible to estimate the size within a $1/\polylog(n)$ factor using only $\polylog(n)$ space \cite{KapralovKS14}; a very recent improvement reduces the polylog factors to $O(\log^2(n))$ \cite{KapralovMNT20}. There is also a line of work that estimates the size of the matching in $o(n)$ space in \emph{adversarial} streams for special classes of graphs as such planar graphs or low-arboricity graphs \cite{EsfandiariHLMO18,BuryS15,McGregorV16,CormodeJMM17,McGregorV18}.

There are many one-pass streaming algorithms for computing a maximum matching in \emph{weighted} graphs. For adversarial-order streaming, a long line of work culminated in a $(1/2-\eps)$-approximation using $O(n)$ space \cite{FeigenbaumKMSZ05,McGregor05,EpsteinLMS11,CrouchS14,PazS17,GhaffariW19}. Gamlath et al. recently showed that for random-order streams, one can achieve an approximation ratio of $1/2 + \Omega(1)$. \cite{GamlathKMS19}. See also other related work on weighted graphs in \cite{BuryS15}.

There are several results on upper and lower bounds for computing a maximum matching in dynamic streams (where edges can also be deleted) \cite{Konrad15,ChitnisCHM15,AssadiKLY16,ChitnisCEHMMV16,AssadiKL17}. Finally, there are several results that are able to achieve better bounds by allowing the algorithm to make multiple passes over the stream: some results focus on just two or three passes \cite{KonradMM12,EsfandiariHM16,KaleT17,Konrad18}, while others seek to compute a $(1-\eps)$-approximate matching by allowing a large constant number (or even $\log(n)$) passes  \cite{McGregor05,AhnG13,GuruswamiO13,AhnG15}.
\paragraph*{Overview of Techniques}
The basic greedy algorithm trivially achieves a $1/2$-approximate matching in adversarial streams; in fact, Konrad et al later showed that the ratio remains $1/2 + o(1)$ even in random-order streams \cite{KonradMM12}. Existing algorithms for improving the $1/2$ ratio in random-order streams generally fall into two categories. The algorithms in \cite{KonradMM12,Konrad18,GamlathKMS19,FarhadiHMRR20} use the randomness of the stream to compute some fraction of short augmenting paths, thus going beyond the $1/2$-approximation of a maximal matching. The result in \cite{AssadiBBMS19} instead shows that one can obtain a large matching by constructing a subgraph that obeys certain degree-properties.

Our result follows the framework of \cite{AssadiBBMS19}. Given any graph $G$, an earlier result of Bernstein and Stein for fully dynamic matching defined the notion of an edge-degree constrainted subgraph (denoted EDCS), which is a sparse subgraph $H \subseteq G$ that obeys certain degree-properties \cite{BernsteinS15}. They showed that any EDCS $H$ always contains a ($2/3-\eps)$-approximate matching. The streaming result of Assadi et al. \cite{AssadiBBMS19} then showed that given a random-order stream, it is possible to compute an EDCS $H$ in $O(n^{1.5})$ space; returning the maximum matching in $H$ yields a $(2/3-\eps)$-approximate matching in $G$.

Our result also takes the EDCS as its starting point, but it is unclear how to compute an EDCS $H$ of $G$ using less than $O(n^{1.5})$ space. Our algorithm requires two new contributions. Firstly, we show that it is sufficient for $H$ to satisfy a somewhat relaxed set of properties. Our main contribution is then to use an entirely different construction of this relaxed subgraph, which uses the randomness of the stream more aggressively to compute $H$ using low space.

\section{Notation and Preliminaries}

Consider any graph $H = (V_H,E_H)$. We define $\deg_H(v)$ to be the degree of $v$ in $H$ and we define the degree of an edge $(u,v)$ to be $\deg_H(u) + \deg_H(v)$. A matching $M$ in $H$ is a set of vertex-disjoint edges. All graphs in this paper are unweighted and undirected. We use $\mu(H)$ to denote the size of the maximum matching in $H$. Unless otherwise indicated, we let $G = (V,E)$ refer to the input graph and let $n = |V|$ and $m = |E|$. We note that every graph referred to in the paper has the same vertex $V$ as the input graph; when we refer to subgraphs, we are always referring to a subset of edges on this same vertex set.

The input graph $G = (V,E)$ is given as a stream of edges $S = \langle e_1, \ldots, e_m \rangle$. We assume that the permutation $(e_1,\ldots, e_m)$ of the edges is chosen \emph{uniformly at random} among all permutations of $E$. We use $\stream{i}{j}$ to denote the substream $\langle e_i, \ldots, e_j \rangle$, and we use $G_{>i} \subseteq G$ to denote the subgraph of $G$ containing all edges in $\{e_{i+1},\ldots,e_m\}$.

Our analysis will apply concentration bounds to segments $\stream{i}{j}$ of the stream. Observe that because the stream is a random permutation, any segment $\stream{i}{j}$ is equivalent to sampling $j - 1 + 1$ edges from the stream without replacement. We can thus apply the Chernoff bound for negatively associated variables (see e.g. the primer in \cite{wajc}).

\begin{theorem}[Chernoff]
\label{thm:chernoff}
Let $X_1, \ldots X_n$ be negatively associated random variables taking values in $[0,1]$. Let $X = \sum X_i$ and let $\mu = \ex[X]$. Then, for any $0 < \delta < 1$ we have 
$$\Pr[X \leq \mu (1-\delta)] \leq \exp(\frac{-\mu \cdot \delta^2}{2}),$$ 
and
$$\Pr[X \geq \mu (1+\delta)] \leq \exp(\frac{-\mu \cdot \delta^2}{3})$$
\end{theorem}

\paragraph*{The early and late sections of the stream}

Our algorithm will use the first $\eps m$ edges of the stream to learn about the graph and will effectively ignore them for the purposes of analyzing the maximum matching. Thus, we only approximate the maximum matching in the later $(1-\eps) m$ edges of stream; because the stream is random, these edges still contain a large fraction of the maximum matching. We use the following definitions and lemmas to formalize this intuition. 

\begin{definition} 
We Let $\eearly$ denote the first $\eps m$ edges of the stream, and $\elate$ denote the rest: that is, $\eearly = \{e_1, \ldots, e_{\eps m}\}$, and  $\elate = \{e_{\eps m + 1}, \ldots, e_m \}$. Define $\gearly = (V, \eearly)$ and $\glate = (V, \elate) = G_{>\eps m}.$
\end{definition}

For the probability bounds to work out, we need to assume that $\mu(G) \geq 20\log(n)\eps^{-2}$. We justify this assumption by observing that every graph $G$ satisfies $m \leq 2n\mu(G)$, so  if $\mu(G) < 20\log(n)\eps^{-2}$, then the algorithm can trivially return an exact maximum matching by simply storing every edge using only $O(m) = O(\log(n)\eps^{-2})$ space. This justifies the following:

\begin{claim} [Assumption]
\label{assumption}
We can assume for the rest of the paper that $\mu(G) \geq 20\log(n)\eps^{-2}$. 
\end{claim}

Combining Claim \ref{assumption} with Chernoff bound we get the following lemma, which allows us to focus our analysis on the edges in $\glate$. 

\begin{lemma} 
\label{lem:glate-matching}
Assuming that $\eps < 1/2$, we have that $\Pr[\mu(\glate) \geq (1-2\eps)\mu(G)] \geq 1 - n^{-5}$.
\end{lemma}

\begin{proof}
Fix some maximum matching $M = (f_1, ..., f_{\mu(G)})$ of $G$. Define $X_i$ to be the indicator variable that edge $f_i \in M$ appears in $\glate$. Since the stream is random, and since $\glate$ contains exactly $(1-\eps)m$ edges, we have that $\mathbb{E}[X_i] = (1-\eps)$ and $\sum \mathbb{E}[X_i] = (1-\eps)\mu(G)$. It is also easy to see that the $X_i$ are negatively associated, since these variables correspond to sampling $(1-\eps)m$ edges without replacement. Recall from Claim \ref{assumption} that we assume $\mu(G) \geq 20\log(n)\eps^{-2}$. Applying the Chernoff Bound in Theorem \ref{thm:chernoff} completes the proof.
\end{proof}

\paragraph*{Existing Work on EDCS}

We now review the basic facts about the edge-degree constrained subgraph (EDCS), which was first introduced in \cite{BernsteinS15}. 

\begin{definition}
Let $G = (V,E)$ be a graph, and $H = (V, E_H)$ a subgraph of $G$. Given any parameters $\beta \geq 2$ and $\lambda < 1$, we say that $H$ is a $(\beta, \lambda)$-EDCS of $G$ if $H$ satisfies the following properties: 
\begin{itemize}
	\item \indent \indent [Property P1:] For any edge $(u,v) \in H$, $\deg_H(u) + \deg_H(v) \leq \beta$
	\item \indent \indent [Property P2:] For any edge $(u,v) \in G \setminus H$, $\deg_H(u) + \deg_H(v) \geq \beta(1-\lambda)$.
\end{itemize}
\end{definition}

The crucial fact about the EDCS is that it always contains a (almost) $2/3$-approximate matching. The simplest proof of Lemma \ref{lem:edcs-matching} below is in Lemma 3.2 of \cite{AssadiB19}.

\begin{lemma} [\cite{AssadiB19}]\label{lem:edcs-matching}
	Let $G(V,E)$ be any graph and $\eps < 1/2$ be some parameter. Let $\lambda, \beta$ be parameters with $\lambda \leq \frac{\eps}{64}$, $\beta \geq 8\lambda^{-2}\log{(1/\lambda)}$. Then, for any $(\beta, \lambda)-EDCS$ $H$ of $G$, we have that $\mu(H) \geq (\frac{2}{3} - \eps) \mu(G)$. (Note that the final guarantee is stated slightly differently than in Lemma 3.2 of \cite{AssadiB19}, and to ensure the two are equivalent, we set $\lambda$ to be a factor of two smaller than in Lemma 3.2 of \cite{AssadiB19}.)
\end{lemma}

\section{Our Modified Subgraph}
Unlike the algorithm of \cite{AssadiBBMS19}, we do not actually construct an EDCS of $G$, as we do not know how to do this in less than $O(n^{1.5})$ space. We instead rely on a more relaxed set of properties, which we analyze using Lemma \ref{lem:edcs-matching} as a black-box. We now introduce some of the basic new tools used by our algorithm. Note that graph $G$ in the lemma and definitions below crucially refers to any arbitrary graph $G$, and not necessarily the main input graph of the streaming algorithm. 

\begin{definition}
We say that a graph $H$ has bounded edge-degree $\beta$ if for every edge $(u,v) \in H$, $\deg_H(u) + \deg_H(v) \leq \beta$. 
\end{definition}

\begin{definition} 
Let $G$ be any graph, and let $H$ be a subgraph of $G$ with bounded edge-degree $\beta$. For any parameter $\lambda < 1$, we say that an edge $(u,v) \in G \setminus H$ is $(G,H,\beta,\lambda)$-underfull if $\deg_H(u) + \deg_H(v) < \beta(1-\lambda)$
\end{definition}

The two definitions above effectively separate the two EDCS properties: any subgraph $H$ of $G$ with bounded edge-degree $\beta$ automatically satisfies property P1 of an EDCS, and underfull edges are then those that violate property P2. We now show that one can always construct a large matching from the combination of these two parts. 

\begin{lemma}
\label{lem:underfull-matching}
Let $\eps < 1/2$ be any parameter, and let $\lambda, \beta$ be parameters with
$\lambda \leq \frac{\eps}{128}$, $\beta \geq 16\lambda^{-2}\log{(1/\lambda)}$.
Consider any graph $G$, and any subgraph $H$ with bounded edge-degree $\beta$. 
Let X contain all edges in $G \setminus H$ that are $(G,H,\beta,\lambda)$-underfull. Then $\mu(X \cup H) \geq (2/3 - \eps) \mu(G)$
\end{lemma}

\begin{proof}
	\newcommand{\mgh}{M^H_G}
	\newcommand{\mgminus}{M^{G \setminus H}_G}
Note that it is NOT necessarily the case that $H \cup X$ is an EDCS of $G$, because adding the edges of $X$ to $H$ will increase vertex and edge degrees in $H$, so $H \cup X$ might not satisfy property P1 of an EDCS. We thus need a more careful argument. 

Let $M_G$ be the maximum matching in $G$, let $\mgh = M_G \cap H$ and let $\mgminus = M_G \cap (G \setminus H)$. Let $X^M = X \cap \mgminus$. Note that by construction, $M_G \subseteq H \cup \mgminus$, so $\mu(H \cup \mgminus)= \mu(G)$.

We now complete the proof by showing that $H \cup X^M$ is a $(\beta + 2, 2\lambda)$-EDCS of $H \cup \mgminus$. Let us start by showing property P2. Recall that $X$ contains all edges $(u,v)$ in $G \setminus H$ for which $\deg_H(u) + \deg_H(v) < \beta(1-\lambda)$, so by construction $X^M$ contains all such edges in $\mgminus$. Thus, every edge $(u,v) \in (H \cup \mgminus) \setminus (H \cup X^M) = \mgminus \setminus X^M$ must have $\deg_H(u) + \deg_H(v) \geq \beta (1-\lambda) \geq (\beta + 2)(1-2\lambda)$, where the last inequality is just rearranging the algebra to fit Property P2 for our new EDCS parameters of $\beta + 2, 2\lambda$.

For property P1, note that $X^M \subseteq \mgminus$ is a matching, so for every vertex $v$ we have $\deg_H(v) \leq \deg_{H \cup X^M}(v) \leq \deg_H(v) + 1$. Now, for $(u,v) \in H$ we had $\deg_H(u) + \deg_H(v) \leq \beta$ (by property P1 of $H$), and for $(u,v) \in X^M \subseteq X$ we had $\deg_H(u) + \deg_H(v) < \beta$ (by definition of $X$). Thus, for every $(u,v) \in H \cup X^M$ we have that $\deg_{H \cup X^M}(u) + \deg_{H \cup X^M}(v) \leq \deg_H(u) + \deg_H(v) + 2 \leq \beta + 2$.

Note that because of how we set the parameters, $\beta' = \beta + 2 < 2\beta$ and $\lambda' = 2\lambda$ satisfy the requirements of Lemma \ref{lem:edcs-matching}. We thus have that $\mu(H \cup X) \geq \mu(H \cup X^M) \geq (2/3 - \eps)\mu(H \cup \mgminus) = (2/3-\eps)\mu(G)$.

\end{proof}

\section{The Algorithm}

\subsection{The Two Phases} Our algorithm will proceed in two phases. Once phase I terminates, the algorithm proceeds to phase II and never returns to phase I. The goal of phase I is to construct a suitable subgraph $H$ of $G$. We now state the formal properties that will be guaranteed by phase I.

\begin{definition}[parameters] \label{dfn:parameters} Throughout this section we use the following parameters. Let $\eps < 1/2$ be the final approximation parameter we are aiming for. Set $\lambda = \frac{\eps}{128}$ and set $\beta = 16\lambda^{-2}\log{(1/\lambda)}$; note that $\lambda$ and $\beta$ are $O(\textrm{poly}(1/\eps))$. Set $\alpha = \frac{\eps m}{n \beta^2 + 1} = O(\frac{m}{n}\textrm{poly}(1/\eps))$ and $\gamma = 5\log(n) \frac{m}{\alpha} = O(n\log(n)\textrm{poly}(1/\eps))$.
\end{definition}

\begin{lemma} \label{lem:phase1}
Phase I uses $O(n\beta) = O(n\textrm{poly}(1/\eps))$ space and constructs a subgraph $H$ of $G$. The phase satisfies the following properties:
\begin{enumerate}
\item \label{phase1-prop1} Phase I terminates within the first $\eps m$ edges of the stream. That is, Phase I terminates at the end of processing some edge $e_i$ with $i \leq \eps m$.
\item \label{phase1-prop2} When Phase I terminates at the end of processing some edge $e_i$, the subgraph $H \subseteq G$ constructed during this phase satisfies the following properties:
\begin{enumerate}
	\item \label{phase1-prop2a} $H$ has bounded edge-degree $\beta$. As a corollary, $H$ has $O(n\beta)$ edges.
	\item \label{phase1-prop2b}  With probability at least $1 - n^{-3}$, the total number of $(G_{>i},H,\beta,\lambda)$-underfull edges in $G_{>i} \setminus H$ is at most $\gamma$. (Recall that $G_{> i}$ denotes the subgraph of $G$ that contains all edges in $\{e_{i+1}, \ldots, e_m\}$.)
\end{enumerate}
\end{enumerate}
\end{lemma}

\noindent We now show that if we can ensure the properties of Lemma \ref{lem:phase1}, our main result follows. 

\begin{proof} [Proof of Theorem \ref{thm:main}]
Let us say that Phase I terminates after edge $e_i$ and let $H$ be the subgraph constructed by Phase I. Phase II of the algorithm proceeds as follows. It initializes an empty set $X$. Then, for every edge $(u,v)$ in $\stream{i+1}{m}$, if $\deg_H(u) + \deg_H(v) < \beta(1-\lambda)$ (that is, if $(u,v)$ is $(G_{>i},H,\beta,\lambda)$-underfull), the algorithm adds edge $(u,v)$ to $X$. After the algorithm completes the stream, it then returns the maximum matching in $H \cup X$.

Let us now analyze the approximation ratio. By property \ref{phase1-prop1} of Lemma \ref{lem:phase1}, $G_{>i} \subseteq \glate$; thus, $X$ contains all $(\glate, H, \beta, \lambda)$-underfull edges. By property \ref{phase1-prop2a}, $H$ has bounded edge-degree $\beta$. Thus, applying Lemma \ref{lem:underfull-matching}, we have that $\mu(H) \geq (2/3 - \eps)\mu(\glate)$. Combining this with Lemma \ref{lem:glate-matching}, we get that $\mu(H) \geq (2/3-\eps)(1-2\eps)\mu(G) \geq (2/3 - 3\eps) \mu(G)$; using $\eps' = \eps/3$ thus yields the desired approximation ratio.

For the space analysis, we know from Lemma \ref{lem:phase1} that Phase I requires $O(n\beta)$ space, which is the space needed to store subgraph $H$. By Property \ref{phase1-prop2b}, the size of $X$ in Phase II is at most $O(n\log(n))$. The overall space is thus $O(n\log(n) + n\beta) = O(n\log(n) + n\textrm{poly}(1/\eps))$. 

Finally, note that the only two probabilistic claims are Lemma \ref{lem:glate-matching} and Property \ref{phase1-prop2b} of Lemma \ref{lem:phase1}, both of which hold with probability $\geq 1 - n^{-3}$. A union bound thus yields an overall probability of success $\geq 1 - 2n^{-3}$.
\end{proof}

\subsection{Decription of Phase I}
All we have left is to describe Phase I and prove Lemma \ref{lem:phase1}. See Algorithm \ref{alg:main} for pseudocode of the entire algorithm. Recall the parameters $\eps, \beta, \lambda, \alpha, \gamma$ from Definition \ref{dfn:parameters}. Phase I is split into epochs, each containing exactly $\alpha$ edges from the stream. So in epoch i, the algorithm looks at $\stream{(i-1)\alpha + 1}{i\alpha}$. 


Phase I initializes the graph $H = \emptyset$. In epoch i, the algorithm goes through the edges of $\stream{(i-1)\alpha + 1}{i\alpha}$ one by one. For edge $(u,v)$, if $\deg_H(u) + \deg_H(v) < (1-\lambda)\beta$, then the algorithm adds edge $(u,v)$ to $H$ (Line 5). (Note that the algorithm changes $H$ over time, so $\deg_H(u) + \deg_H(v)$ always refers to the degrees in $H$ at the time edge $(u,v)$ is being examined.) After each edge insertion to $H$, the algorithm runs procedure {\em RemoveOverfullEdges($H$)} (Line 7); this procedure repeatedly picks an edge $(x,y)$ with $\deg_H(x) + \deg_H(y) > \beta$ until no such edge remains. Note that as a result, our algorithm preserves the invariant that $H$ always has bounded edge-degree $\beta$.

In each epoch, the algorithm also has a single boolean \foundu, which is set to \true\ if the algorithm ever adds an edge to $H$ during that epoch. At the end of the epoch, if \foundu\ is set to \true, then the algorithm simply proceeds to the next epoch. If \foundu\ is \false, then the algorithm permanently terminates Phase I and proceeds to Phase II. (The intuition is that since the ordering of the stream is random, if the algorithm failed to find an underfull edge in an entire epoch, then there must be relatively few underfull edges left in the stream, so Property \ref{phase1-prop2b} of Lemma \ref{lem:phase1} will be satisfied.)

Note that \foundu\ being false is the only way Phase I can terminate (Line 9); we prove in the analysis that this deterministically occurs within the first $\eps m$ edges of the stream.

\begin{algorithm2e}
		
	\caption{The algorithm for computing a matching in a random-order stream. After initialization, the algorithm goes to Phase I. Once the algorithm exits Phase I, it moves on to Phase II and never returns to Phase I. Line 9 is the only place where the algorithm can exit Phase I.}
	\label{alg:main}

	\SetAlgoSkip{}
	\SetKwProg{procedure}{Procedure}{}{}
	\SetKwBlock{DoUntilTermination}{Do Until Termination}

\procedure{Initilization}{
			Initialize $H = \emptyset$ \tcc*[f]{$H$ is a global variable modified by Phase I} 
			
			Let $\eps < 1/2$ be the main approximation parameter 
			
			Set $\lambda = \frac{\eps}{128}, \ \beta = 16\lambda^{-2}\log{(1/\lambda)}, \ \alpha = \frac{\eps m}{n \beta^2+1}, \ \gamma = 5\log(n) \frac{m}{\alpha}$ (Definition \ref{dfn:parameters}). 
			
			{\bf Go To} Phase I 
}	

\procedure{Phase I}{
	
		\DoUntilTermination(\tcc*[f]{each iteration corresponds to one epoch}) {	
		(1) \foundu $\ \leftarrow \ $ FALSE 
		
		(2) \For(\tcc*[f]{each epoch looks at exactly $\alpha$ edges.}){$\alpha$ Iterations:}  {
			
		(3) Let $(u,v)$ be the next edge in the stream 
		
		(4) \If{$\deg_H(u) + \deg_H(v) < \beta(1-\lambda)$} {
			
		(5) Add edge $(u,v)$ to $H$  \tcc*[f]{note: this increases $\deg_H(u)$ and $\deg_H(v)$.} 
		
		(6) \foundu $\ \leftarrow \ $ TRUE 
		
		(7) {\em RemoveOverfullEdges($H$)} 
	} 
} 
	
		(8) \If{\foundu\ = FALSE} {
			
		(9) {\bf Go To} Phase II \tcc*[f]{permanently exit Phase I.} \;
		
	} 

\tcc{Else, will move on to the next epoch of Phase I.}
		
} 
} 

\procedure{RemoveOverfullEdges($H$)} {
		(1) \While{there exists $(u,v) \in H$ such that $\deg_H(u) + \deg_H(v) > \beta$}{
			
		(2) Remove $(u,v)$ from $H$ $\qquad$ \tcc*[f]{note: this decreases $\deg_H(u)$ and $\deg_H(v)$} 
	} 
		
		\tcc*[f]{note: when the while loop terminates, $H$ is guaranteed to have bounded edge-degree $\beta$.}
} 

\procedure{Phase II} {
		(1) Initialize $X \leftarrow \emptyset$ $\qquad$ \tcc*[f]{all underfull edges will be added to $X$} \;
		(2) \ForEach{remaining edge $(u,v)$ in the stream} {
		(3) \If{$\deg_H(u) + \deg_H(v) < \beta(1-\lambda)$} {
		(4) Add edge $(u,v)$ to $X$ $\quad$  \tcc*[f]{note: this does  \emph{NOT} change any $\deg_H(v)$.} 
	} 
} 
	()	(5) {\bf Return} the maximum matching in $H \cup X$ \;
} 
	
\end{algorithm2e}

\subsection{Analysis}
We now turn to proving Lemma \ref{lem:phase1}. The hardest part is proving Property \ref{phase1-prop1}. Observe that every epoch that doesn't terminate Phase I must add at least one edge to $H$. To prove Property \ref{phase1-prop1}, we use an auxiliary lemma that bounds the total number of changes made to $H$. 

\begin{lemma}
\label{lem:edcs-changes}
Fix any parameter $\beta > 2$. Let $H = (V_H, E_H)$ be a graph, with $E_H$ initially empty. Say that an adersary adds and removes edges from $H$ using an arbitrary sequence of two possible moves
\begin{itemize}
	\item \indent \indent [Deletion Move] Remove an edge $(u,v)$ from $H$ for which $\deg_H(u) + \deg_H(v) > \beta$
	\item \indent \indent [Insertion Move] Add an edge $(u,v)$ to $H$ for some pair $u,v \in V$ for which $\deg_H(u) + \deg_H(v) < \beta - 1$.  
\end{itemize}
Then, after $n\beta^2$ moves, no legal move remains.
\end{lemma}

\begin{proof}
The proof is similar to that of Proposition 2.4 in \cite{AssadiB19}. Define the following potential functions $\Phi_1(H) = (\beta - 1/2) \cdot \sum_{v \in V_H} \deg_H(v)$, $\Phi_2(H) = \sum_{(u,v) \in E_H} \deg_H(u) + \deg_H(v)$, and the main potential function $\Phi(H) = \Phi_1(H) - \Phi_2(H)$. Note that initially $H$ is empty so $\Phi(H) = 0$. We claim that at all times $\Phi(H) \leq \Phi_1(H) \leq n\beta^2$. To see this, note that every vertex $v \in V_H$ always has $\deg_H(v) \leq \beta$, because as long as $\deg_H(v) = \beta$, the adversary cannot perform any insertion moves incident to $v$. In the rest of the proof, we show that every Insertion/Deletion move increases $\Phi(H)$ by at least 1; combined with the fact that at all times $0 \leq \Phi(H) \leq n\beta^2$, we get that there are at most $n\beta^2$ moves in total.

Consider any Deletion Move of edge $(u,v)$. Clearly $\Phi_1(v)$ decreases by exactly $2\beta - 1$. We now show that $\Phi_2(v)$ decreases by at least $2\beta$. One the one hand, $\Phi_2(v)$ decreases by at least $\beta + 1$ because edge $(u,v)$ no longer participates in the sum, and $\deg_H(u) + \deg_H(v)$ was $> \beta$ before the deletion. But at the same time, since $\deg_H(u) + \deg_H(v) \geq \beta + 1$ before the deletion, there are at least $\beta - 1$ edges other than $(u,v)$ incident to $u$ or $v$, and each of their edge degrees decrease by 1 in the sum for $\Phi_2(H)$. Thus, $\Phi_2(H)$ decreases by at least $\beta + 1 + (\beta - 1) = 2\beta$, while $\Phi_1(H)$ decreases by exactly $2\beta - 1$, so overall $\Phi(H) = \Phi_1(H) - \Phi_2(H)$ increases by at least one.

Similarly, consider any Insertion Move of edge $(u,v)$. Clearly $\Phi_1(v)$ increases by exactly $2\beta - 1$. We now show that $\Phi_2(v)$ increases by at most $2\beta-2$. Recall that $\deg_H(u) + \deg_H(v) \leq \beta - 2$ before the insertion, so after the insertion we have that $\deg_H(u) + \deg_H(v) \leq \beta$, so the edge $(u,v)$ itself contributes at most $\beta$ to the sum in $\Phi_2$. There are also at most $\beta - 2$ edges other than $(u,v)$ incident to $u$ or $v$, each of whose edge degrees increases by 1. Thus, overall, $\Phi_2(H)$ increases by at most $\beta + (\beta - 2) = 2\beta - 2$, so $\phi(H)$ increases by at least $(2\beta-1) - (2\beta -2) = 1$.
\end{proof}

\begin{proof}[Proof of Lemma \ref{lem:phase1}]
\newcommand{\eventa}{\mathcal{A}}
\newcommand{\eventb}{\mathcal{B}}
\newcommand{\eremain}{E^r}
\newcommand{\eepoch}{E^e}
\newcommand{\ebad}{E^u}
Property \ref{phase1-prop2a} is clearly satisfied by construction, because after any insertion to $H$ the algorithm runs {\em RemoveOnderfullEdges(H)} (line 7) to ensure that $H$ has bounded edge-degree $\beta$. As a result, we clearly have that every \emph{vertex} degree is at most $\beta$, so Phase I needs only $O(n\beta)$ space to store $H$.

For the proof of Property \ref{phase1-prop1}, observe that any changes the algorithm makes to H follow the rules for Insertion/Deletion moves from Lemma \ref{lem:edcs-changes}, so Algorithm \ref{alg:main} makes at most $n\beta^2$ changes to $H$. (Line 5 of Phase I corresponds to deletion moves in Lemma \ref{lem:edcs-changes}, while line 2 of {\em RemoveOverfullEdges($H$)} corresponds to insertion moves. Note that line 5 of phase I actually obeys an even stronger inequality than deletion moves, since $\beta(1-\lambda) < \beta -1$.) Each epoch that does not terminate Phase I makes at least one change to $H$, so phase I goes through at most $n\beta^2 + 1$ epochs before termination. Each epoch contains $\alpha$ edges, so overall Phase I goes through at most $\alpha (n \beta^2 + 1) = \eps m$ edges, as desired.

All that remains is to prove Property \ref{phase1-prop2b}. As mentioned above, the intuition is simple: the algorithm only exits Phase I if it fails to find a single underfull edge in the entire epoch (Line 9), and since the stream is random, such an event implies that there are probably relatively few underfull edges left in the stream. We now formalize this intuition.

Let $\eventa_i$ be the event that \foundu\ is set to FALSE in epoch $i$. Recall that epoch $i$ ends on edge $e_{i\alpha}$; let $\eventb_i$ be the event that the number of $(G_{>i\alpha},H,\beta, \lambda)$-underfull edges is more than $\gamma$. Note that Property \ref{phase1-prop2b} fails to hold if and only if we have $\eventa_i \land \eventb_i$ for some $i$, so we now upper bound $\Pr[A_i \land B_i]$. Our bound relies on the randomness of the stream. Let $\eremain_i$ contain all edges in the graph that have not yet appeared in the stream at the \emph{beginning} of epoch $i$ (r for remaining). Let $\eepoch_i$ be the edges that appear in epoch $i$ (e for epoch), and note that $\eepoch_i$ is a subset of size $\alpha$ chosen uniformly at random from $\eremain_i$. Define $H_i$ to be the subgraph $H$ at the beginning of epoch $i$, and define $\ebad_i \subseteq \eremain_i$ to be the set $\{(u,v) \in \eremain_i \ | \ \deg_{H_i}(u) + \deg_{H_i}(v) < \beta(1-\lambda)\}$ (u for underfull). Observe that because of event $\eventa_i$, the graph $H$ does not change throughout epoch $i$, so an edge that is underfull at any point during the epoch will be underfull at the end as well. Thus, $\eventa_i \land \eventb_i$ is equivalent to the event that $|\ebad_i| > \gamma$ but $\ebad_i \cap \eepoch_i = \emptyset$. 

Let $\eventa^k_i$ be the event that the kth edge of epoch $i$ is not in $\ebad_i$. We have that 
$$\Pr[\eventb_i \land \eventa_i] \leq \Pr[\eventa_i \cond \eventb_i] = \Pr[\eventa^1_i \cond \eventb_i] \prod_{k=2}^\alpha \Pr[\eventa^k_i \cond \eventb_i, \eventa^1_i, \ldots, \eventa^{k-1}_i].$$ 
Now, observe that
$$\Pr[\eventa^1_i \cond \eventb_i] < 1 - \frac{\gamma}{m}$$ 
because the first edge of the epoch is chosen uniformly at random from the set of $\leq m$ remaining edges, and the event fails if the chosen edge is in $\ebad_i$, where $|\ebad_i| > \gamma$ by definition of $\eventb_i$. Similarly, for any $k$, 
$$\Pr[\eventa^k_i \cond \eventb_i, \eventa^1_i, \ldots, \eventa^{k-1}_i] < 1 - \frac{\gamma}{m}$$
because conditioning on the previous events $\eventa^j_i$ implies that no edge from $\ebad_i$ has yet appeared in this epoch, so there are still at least $\gamma$ edges from $\ebad_i$ left in the stream.

Recall from Definition \ref{dfn:parameters} that $\gamma = 5\log(n)\cdot \frac{m}{\alpha}$. Combining the three above equations yields that $\Pr[\eventb_i \land \eventa_i] \leq (1-\frac{\gamma}{m})^{\alpha} = (1 - \frac{5\log(n)}{\alpha})^\alpha \leq n^{-5}$. There are clearly at most $n^2$ epochs, so union bounding over all of them shows that Property \ref{phase1-prop2b} fails with probability at most $n^{-3}$, as desired.

\end{proof}

\section{Open Problems}
We presented a new single-pass streaming algorithm for computing a maximum matching in a \emph{random-order} stream. The algorithm achieves a $(2/3-\eps)$-approximation using $O(n\log(n))$ space; these bounds improve upon all previous results for the problem.

But while $2/3$ is a natural boundary, there is no reason to believe it is the best possible. Is there an algorithm with approximation ratio $2/3 + \Omega(1)$? Is it possible to compute a $(1-\eps)$-approximate matching in random-order streams? A lower bound of $1 - \Omega(1)$ in this setting would also be extremely interesting.

Another natural open problem is get improved bounds for weighted graphs. Gamlath et al \cite{GamlathKMS19} recently broke through the barrier of $1/2$ and presented an algorithm for weighted graphs that computes a $.506$-approximation (or $.512$ in bipartite graphs) in random-order streams. Can we improve the approximation ratio to $2/3$ in weighted graphs? To $(1-\eps)$?

\section{Acknowledgments} I want to thank Sepehr Assadi for several very helpful discussions.

\bibliography{general}

\begin{thebibliography}{10}

\bibitem{AhnG13}
K.~J. Ahn and S.~Guha.
\newblock Linear programming in the semi-streaming model with application to
  the maximum matching problem.
\newblock {\em Inf. Comput.}, 222:59--79, 2013.

\bibitem{AhnG15}
K.~J. Ahn and S.~Guha.
\newblock Access to data and number of iterations: Dual primal algorithms for
  maximum matching under resource constraints.
\newblock In {\em Proceedings of the 27th {ACM} on Symposium on Parallelism in
  Algorithms and Architectures, {SPAA} 2015, Portland, OR, USA, June 13-15,
  2015}, pages 202--211, 2015.

\bibitem{AssadiBBMS19}
S.~Assadi, M.~Bateni, A.~Bernstein, V.~S. Mirrokni, and C.~Stein.
\newblock Coresets meet {EDCS:} algorithms for matching and vertex cover on
  massive graphs.
\newblock In {\em Proceedings of the Thirtieth Annual {ACM-SIAM} Symposium on
  Discrete Algorithms, {SODA} 2019, San Diego, California, USA, January 6-9,
  2019}, pages 1616--1635, 2019.

\bibitem{AssadiB19}
S.~Assadi and A.~Bernstein.
\newblock Towards a unified theory of sparsification for matching problems.
\newblock In {\em 2nd Symposium on Simplicity in Algorithms, SOSA@SODA 2019,
  January 8-9, 2019 - San Diego, CA, {USA}}, pages 11:1--11:20, 2019.

\bibitem{AssadiKL17}
S.~Assadi, S.~Khanna, and Y.~Li.
\newblock On estimating maximum matching size in graph streams.
\newblock In {\em Proceedings of the Twenty-Eighth Annual {ACM-SIAM} Symposium
  on Discrete Algorithms, {SODA} 2017, Barcelona, Spain, Hotel Porta Fira,
  January 16-19}, pages 1723--1742, 2017.

\bibitem{AssadiKLY16}
S.~Assadi, S.~Khanna, Y.~Li, and G.~Yaroslavtsev.
\newblock Maximum matchings in dynamic graph streams and the simultaneous
  communication model.
\newblock In R.~Krauthgamer, editor, {\em Proceedings of the Twenty-Seventh
  Annual {ACM-SIAM} Symposium on Discrete Algorithms, {SODA} 2016, Arlington,
  VA, USA, January 10-12, 2016}, pages 1345--1364. {SIAM}, 2016.

\bibitem{BernsteinS15}
A.~Bernstein and C.~Stein.
\newblock Fully dynamic matching in bipartite graphs.
\newblock In {\em Automata, Languages, and Programming - 42nd International
  Colloquium, {ICALP} 2015, Kyoto, Japan, July 6-10, 2015, Proceedings, Part
  {I}}, pages 167--179, 2015.

\bibitem{BuryS15}
M.~Bury and C.~Schwiegelshohn.
\newblock Sublinear estimation of weighted matchings in dynamic data streams.
\newblock In {\em Algorithms - {ESA} 2015 - 23rd Annual European Symposium,
  Patras, Greece, September 14-16, 2015, Proceedings}, pages 263--274, 2015.

\bibitem{ChitnisCEHMMV16}
R.~Chitnis, G.~Cormode, H.~Esfandiari, M.~Hajiaghayi, A.~McGregor,
  M.~Monemizadeh, and S.~Vorotnikova.
\newblock Kernelization via sampling with applications to finding matchings and
  related problems in dynamic graph streams.
\newblock In {\em Proceedings of the Twenty-Seventh Annual {ACM-SIAM} Symposium
  on Discrete Algorithms, {SODA} 2016, Arlington, VA, USA, January 10-12,
  2016}, pages 1326--1344, 2016.

\bibitem{ChitnisCHM15}
R.~H. Chitnis, G.~Cormode, M.~T. Hajiaghayi, and M.~Monemizadeh.
\newblock Parameterized streaming: Maximal matching and vertex cover.
\newblock In {\em Proceedings of the Twenty-Sixth Annual {ACM-SIAM} Symposium
  on Discrete Algorithms, {SODA} 2015, San Diego, CA, USA, January 4-6, 2015},
  pages 1234--1251, 2015.

\bibitem{CormodeJMM17}
G.~Cormode, H.~Jowhari, M.~Monemizadeh, and S.~Muthukrishnan.
\newblock The sparse awakens: Streaming algorithms for matching size estimation
  in sparse graphs.
\newblock In {\em 25th Annual European Symposium on Algorithms, {ESA} 2017,
  September 4-6, 2017, Vienna, Austria}, pages 29:1--29:15, 2017.

\bibitem{CrouchS14}
M.~Crouch and D.~S. Stubbs.
\newblock Improved streaming algorithms for weighted matching, via unweighted
  matching.
\newblock In {\em APPROX/RANDOM 2014}, pages 96--104, 2014.

\bibitem{EpsteinLMS11}
L.~Epstein, A.~Levin, J.~Mestre, and D.~Segev.
\newblock Improved approximation guarantees for weighted matching in the
  semi-streaming model.
\newblock {\em {SIAM} J. Discrete Math.}, 25(3):1251--1265, 2011.

\bibitem{EsfandiariHLMO18}
H.~Esfandiari, M.~Hajiaghayi, V.~Liaghat, M.~Monemizadeh, and K.~Onak.
\newblock Streaming algorithms for estimating the matching size in planar
  graphs and beyond.
\newblock {\em {ACM} Trans. Algorithms}, 14(4):48:1--48:23, 2018.

\bibitem{EsfandiariHM16}
H.~Esfandiari, M.~Hajiaghayi, and M.~Monemizadeh.
\newblock Finding large matchings in semi-streaming.
\newblock In C.~Domeniconi, F.~Gullo, F.~Bonchi, J.~Domingo{-}Ferrer,
  R.~Baeza{-}Yates, Z.~Zhou, and X.~Wu, editors, {\em {IEEE} International
  Conference on Data Mining Workshops, {ICDM} Workshops 2016, December 12-15,
  2016, Barcelona, Spain}, pages 608--614. {IEEE} Computer Society, 2016.

\bibitem{FarhadiHMRR20}
A.~Farhadi, M.~T. Hajiaghayi, T.~Mai, A.~Rao, and R.~A. Rossi.
\newblock Approximate maximum matching in random streams.
\newblock In {\em Proceedings of the 2020 {ACM-SIAM} Symposium on Discrete
  Algorithms, {SODA} 2020, Salt Lake City, UT, USA, January 5-8, 2020}, pages
  1773--1785, 2020.

\bibitem{FeigenbaumKMSZ05}
J.~Feigenbaum, S.~Kannan, A.~McGregor, S.~Suri, and J.~Zhang.
\newblock On graph problems in a semi-streaming model.
\newblock {\em Theor. Comput. Sci.}, 348(2-3):207--216, 2005.

\bibitem{GamlathKMS19}
B.~Gamlath, S.~Kale, S.~Mitrovic, and O.~Svensson.
\newblock Weighted matchings via unweighted augmentations.
\newblock In P.~Robinson and F.~Ellen, editors, {\em Proceedings of the 2019
  {ACM} Symposium on Principles of Distributed Computing, {PODC} 2019, Toronto,
  ON, Canada, July 29 - August 2, 2019}, pages 491--500. {ACM}, 2019.

\bibitem{GhaffariW19}
M.~Ghaffari and D.~Wajc.
\newblock Simplified and space-optimal semi-streaming (2+epsilon)-approximate
  matching.
\newblock In J.~T. Fineman and M.~Mitzenmacher, editors, {\em 2nd Symposium on
  Simplicity in Algorithms, SOSA@SODA 2019, January 8-9, 2019 - San Diego, CA,
  {USA}}, volume~69 of {\em {OASICS}}, pages 13:1--13:8. Schloss Dagstuhl -
  Leibniz-Zentrum fuer Informatik, 2019.

\bibitem{GoelKK12}
A.~Goel, M.~Kapralov, and S.~Khanna.
\newblock On the communication and streaming complexity of maximum bipartite
  matching.
\newblock In {\em Proceedings of the Twenty-third Annual ACM-SIAM Symposium on
  Discrete Algorithms}, SODA '12, pages 468--485. SIAM, 2012.

\bibitem{GuruswamiO13}
V.~Guruswami and K.~Onak.
\newblock Superlinear lower bounds for multipass graph processing.
\newblock In {\em Proceedings of the 28th Conference on Computational
  Complexity, {CCC} 2013, K.lo Alto, California, USA, 5-7 June, 2013}, pages
  287--298, 2013.

\bibitem{KaleT17}
S.~Kale and S.~Tirodkar.
\newblock Maximum matching in two, three, and a few more passes over graph
  streams.
\newblock In K.~Jansen, J.~D.~P. Rolim, D.~Williamson, and S.~S. Vempala,
  editors, {\em Approximation, Randomization, and Combinatorial Optimization.
  Algorithms and Techniques, {APPROX/RANDOM} 2017, August 16-18, 2017,
  Berkeley, CA, {USA}}, 2017.

\bibitem{Kapralov13}
M.~Kapralov.
\newblock Better bounds for matchings in the streaming model.
\newblock In {\em Proceedings of the Twenty-Fourth Annual {ACM-SIAM} Symposium
  on Discrete Algorithms, {SODA} 2013, New Orleans, Louisiana, USA, January
  6-8, 2013}, pages 1679--1697, 2013.

\bibitem{KapralovKS14}
M.~Kapralov, S.~Khanna, and M.~Sudan.
\newblock Approximating matching size from random streams.
\newblock In {\em Proceedings of the Twenty-Fifth Annual {ACM-SIAM} Symposium
  on Discrete Algorithms, {SODA} 2014, Portland, Oregon, USA, January 5-7,
  2014}, pages 734--751, 2014.

\bibitem{KapralovMNT20}
M.~Kapralov, S.~Mitrovic, A.~Norouzi{-}Fard, and J.~Tardos.
\newblock Space efficient approximation to maximum matching size from uniform
  edge samples.
\newblock In {\em Proceedings of the 2020 {ACM-SIAM} Symposium on Discrete
  Algorithms, {SODA} 2020, Salt Lake City, UT, USA, January 5-8, 2020}, pages
  1753--1772, 2020.

\bibitem{Konrad15}
C.~Konrad.
\newblock Maximum matching in turnstile streams.
\newblock In {\em Algorithms - {ESA} 2015 - 23rd Annual European Symposium,
  Patras, Greece, September 14-16, 2015, Proceedings}, pages 840--852, 2015.

\bibitem{Konrad18}
C.~Konrad.
\newblock A simple augmentation method for matchings with applications to
  streaming algorithms.
\newblock In I.~Potapov, P.~G. Spirakis, and J.~Worrell, editors, {\em 43rd
  International Symposium on Mathematical Foundations of Computer Science,
  {MFCS} 2018, August 27-31, 2018, Liverpool, {UK}}, volume 117 of {\em
  LIPIcs}, pages 74:1--74:16. Schloss Dagstuhl - Leibniz-Zentrum fuer
  Informatik, 2018.

\bibitem{KonradMM12}
C.~Konrad, F.~Magniez, and C.~Mathieu.
\newblock Maximum matching in semi-streaming with few passes.
\newblock In {\em Approximation, Randomization, and Combinatorial Optimization.
  Algorithms and Techniques - 15th International Workshop, {APPROX} 2012, and
  16th International Workshop, {RANDOM} 2012, Cambridge, MA, USA, August 15-17,
  2012. Proceedings}, pages 231--242, 2012.

\bibitem{McGregor05}
A.~McGregor.
\newblock Finding graph matchings in data streams.
\newblock In {\em Approximation, Randomization and Combinatorial Optimization,
  Algorithms and Techniques, 8th International Workshop on Approximation
  Algorithms for Combinatorial Optimization Problems, {APPROX} 2005 and 9th
  InternationalWorkshop on Randomization and Computation, {RANDOM} 2005,
  Berkeley, CA, USA, August 22-24, 2005, Proceedings}, pages 170--181, 2005.

\bibitem{McGregorV16}
A.~McGregor and S.~Vorotnikova.
\newblock Planar matching in streams revisited.
\newblock In {\em Approximation, Randomization, and Combinatorial Optimization.
  Algorithms and Techniques, {APPROX/RANDOM} 2016, September 7-9, 2016, Paris,
  France}, pages 17:1--17:12, 2016.

\bibitem{McGregorV18}
A.~McGregor and S.~Vorotnikova.
\newblock A simple, space-efficient, streaming algorithm for matchings in low
  arboricity graphs.
\newblock In {\em 1st Symposium on Simplicity in Algorithms, {SOSA} 2018,
  January 7-10, 2018, New Orleans, LA, {USA}}, pages 14:1--14:4, 2018.

\bibitem{PazS17}
A.~Paz and G.~Schwartzman.
\newblock A (2 + epsilon)-approximation for maximum weight matching in the
  semi-streaming model.
\newblock In {\em Proceedings of the Twenty-Eighth Annual {ACM-SIAM} Symposium
  on Discrete Algorithms, {SODA} 2017, Barcelona, Spain, Hotel Porta Fira,
  January 16-19}, pages 2153--2161, 2017.

\bibitem{wajc}
D.~Wajc.
\newblock Negative association: definition, properties, and applications.

\end{thebibliography}

\end{document}